\documentclass{ifacconf}

\usepackage{graphicx}      
\usepackage{natbib}        
\usepackage{url} 
\usepackage{algorithm}
\usepackage{algorithmic}
\usepackage{amsmath}
\usepackage{mathabx}
\usepackage{xcolor}

\newcommand{\vMinj}{{v}_\mathit{min,j}} 
\newcommand{\vMaxj}{{v}_\mathit{max,j}}

\newtheorem{theorem}{\textbf{Theorem}}
\newtheorem{probl}{\textbf{Problem}}
\newtheorem{coro}{\textbf{Corollary}}
\newtheorem{rmk}{\textbf{Remark}}

\begin{document}

\thispagestyle{empty}
\onecolumn
{\large
\noindent \copyright\ 2026 the authors. This work has been accepted to IFAC
for publication under a Creative Commons Licence CC-BY-NC-ND.

\noindent This paper has been accepted for presentation at the
23rd IFAC World Congress, 2026.
}
\newpage
\twocolumn
\setcounter{page}{1}

\maketitle

\begin{frontmatter}

\title{Safety-Assured Arrival Scheduling in Sequential UAM Corridor Sections\\under Speed and Separation Constraints\thanksref{footnoteinfo}
} 

 \thanks[footnoteinfo]
 {This work is based on results obtained from a project, JPNP22002, commissioned by the New Energy and Industrial Technology Development Organization (NEDO).}

\author[First]{Sasinee Pruekprasert} 
\author[Second]{Shinji Nakadai} 
\author[First]{Katsuhiro Nishinari}

\address[First]{The University of Tokyo, Tokyo, Japan. (e-mail: spruekprasert@g.ecc.u-tokyo.ac.jp, tknishi@mail.ecc.u-tokyo.ac.jp)}
\address[Second]{Intent Exchange, Inc., Tokyo,  Japan.\\(e-mail: nakadai@intent-exchange.com)}
 
\begin{abstract}  
This paper presents a safety-assured arrival-scheduling framework for Urban Air Mobility (UAM) corridor operations.
We propose an analytical method to compute a sufficient ETA gap at Constrained Waypoints (CWPs) that guarantees
longitudinal separation along sequential corridor sections with heterogeneous speed limits.
The resulting ETA-gap condition depends on section-specific
speed bounds and the required separation distance, providing
an efficiently computable rule suitable for integration into
future digital ETA-scheduling and air traffic management systems.
We show that the computed ETA gap ensures safe separation across all corridor sections under prescribed section travel times and speed limits. 
Numerical simulations for a decreasing-speed corridor
confirm that vehicles coordinated with the proposed
mechanism adjust their speeds to maintain the required
spacing, avoid potential collisions, and support improved
traffic flow compared with unscheduled operations.
\end{abstract}


\end{frontmatter}
\section{Introduction} 

Urban Air Mobility (UAM) introduces new air traffic management concepts for on-demand and scheduled aerial transportation in metropolitan areas, attracting growing attention from industry, academia, and government~\citep{thipphavong2018urban}.
As demand increases, existing Air Traffic Control (ATC) systems may face capacity challenges under high traffic density~\citep{vascik2018scaling}. 
At the same time, structured airspace concepts can improve safety but may increase travel delays compared with free-flight operations~\citep{bauranov2021designing}.
It is therefore essential to develop coordination mechanisms that maintain both safety and operational efficiency under these emerging constraints.

To address these challenges, the Federal Aviation Administration (FAA) incorporated UAM corridors into its UAM ConOps~\citep{fontaine2023urban}. 
UAM corridors are reserved airspace volumes managed by designated authorities, such as Providers of Services for UAM (PSUs), who regulate access and performance. 
These corridors can support applications such as air taxis~\citep{muna2021air}, helping to reduce congestion and infrastructure constraints~\citep{wang2021air}, and are envisioned for electric Vertical Takeoff and Landing (eVTOL) and powered-lift aircraft within the FAA’s Advanced Air Mobility (AAM) framework.

Prior work has studied UAM corridor design through graph-based congestion models~\citep{wang2021air, jiang2022metrics}, geometric or layered architectures~\citep{muna2021air}, and taxonomies of structured airspace trade-offs~\citep{lee2023airspace}. 
\cite{Asslouj2024fixed} further extends these concepts to fixed-corridor systems for Unmanned Aircraft Systems (UAS). However, these studies focus mainly on corridor geometry and routing rather than separation conditions for safe inter-vehicle spacing, as envisioned under Digital Flight~\citep{wing2022digital}.

Conventional airspace-management research has also studied coordinated scheduling and separation assurance, including Integrated Demand Management~\citep{smith2016integrated}, spatiotemporal conflict-free scheduling~\citep{hildum2012scheduling}, human-in-the-loop comparisons of separation responsibility~\citep{wing2010comparison}, and pre-flight scheduling for AAM operations~\citep{yokoyama2025performance}. 
Although these works highlight the value of scheduled-arrival management, they do not address the fine-grained, waypoint-level safety requirements for high-density UAM corridors.

Building on this foundation, this work extends the ETA-based safety framework to sequences of corridor sections, each with distinct speed bounds that may increase or decrease along the route.
Such multi-section configurations naturally arise in realistic UAM operations due to changes in altitude layers, airspace policies, and local traffic densities.
Unlike the single-section setting in~\citep{pruekprasert2025safe}, heterogeneous speed limits introduce new challenges: a safe ETA gap computed for one section may no longer preserve separation after vehicles transition into an adjacent section with different constraints, and maintaining safety must account for the possibility that vehicles occupy different sections at the same time. 
To address these issues, our main contributions are:
\begin{itemize}
    \item extending ETA-based coordination to sequential corridor sections with heterogeneous speed limits;
    \item constructing conservative spatiotemporal bounds that enclose all admissible trajectories under prescribed section travel times and speed limits;
    \item reducing the safe ETA-gap computation to checking a finite set of critical times induced by the trajectory-bound critical points; and
    \item complementing the theoretical results with numerical validation of the computed ETA gaps in a decreasing-speed corridor with downstream traffic compression.
\end{itemize}
Our method provides efficiently computable temporal-separation conditions that reduce the need
for the complex optimization required in~\citep{pruekprasert2025safe} and remain scalable for multi-section UAM ETA management.
The computed ETA gap can therefore be incorporated as a safety constraint in broader ETA-scheduling and traffic-management systems.

The remainder of this paper is organized as follows. 
Section~\ref{sec:UAMcorridor} introduces the UAM corridor setting and formulates the safe ETA-gap problem. 
Section~\ref{sec:bounds} derives conservative spatiotemporal bounds under corridor speed limits, and Section~\ref{sec:etagap} uses these bounds to construct a computationally efficient guaranteed-safe ETA gap. 
Section~\ref{sec:simulation} presents numerical simulations. Section~\ref{sec:conclusion} presents the conclusion. 

\section{UAM Corridor with ETAs at CWPs}
\label{sec:UAMcorridor}  

We consider Air Traffic Control (ATC) for inter-vehicle self-separation in Urban Air Mobility (UAM) corridors using shared Estimated Times of Arrival (ETAs) at Constrained Waypoints (CWPs).
As illustrated in Fig.~\ref{fig:UAM_ETA}, CWPs partition the airspace into multiple corridor sections.
Each section may reserve specific airspace for different flight phases, and CWPs connect these sections to form an integrated UAM corridor network.




 
\subsection{Vehicle Kinematics and Corridor Constraints}


We consider $n$ aerial vehicles, indexed by $i \in \{0,\dots,n-1\}$, entering the corridor in order.
Vehicle motion within is modeled as one-dimensional and governed by 
\begin{equation}\label{eq:kinematic}
    \dot{x}_i(t) = v_i(t), \qquad 
\dot{v}_i(t) = a_i(t),
\end{equation}
where $x_i(t)$ and $v_i(t)$ denote the position and speed of vehicle~$i$, and $a_i(t)$ is its acceleration input at time~$t$.

We consider $m$ sequential corridor sections indexed by $j \in {0,\dots,m-1}$, where each section $j$ connects CWP$_j$ and CWP$_{j+1}$.
For each section~$j$, we specify speed limits $\vMinj$ and $\vMaxj$ representing the operational bounds imposed by the corridor design or airspace regulations. 
Accordingly, while vehicle~$i$ is within section~$j$, its speed must satisfy
\begin{equation}
\label{eq vlim}
0 < \vMinj \le v_i(t) \le \vMaxj.
\end{equation}
We also assume consecutive speed-limit intervals overlap, i.e.,
$[\vMinj,\vMaxj]\cap[{v}_\mathit{min,\,j+1},{v}_\mathit{max,\,j+1}]\neq\emptyset$, so vehicles can adjust speeds when transitioning between sections.

To simplify ETA scheduling, we assign all vehicles a common travel time in each corridor section~$j$ to be
$ 
\tau_j = \frac{l_j}{v_{\mathit{avg},j}},
$ 
where $v_{\mathit{avg},j} = (\vMinj + \vMaxj)/2$ is a representative average speed and $l_j$ is the section length. 
The framework, however, can be extended to heterogeneous travel times. 

\begin{figure}[tp] 
\centering
\includegraphics[width=1\columnwidth]{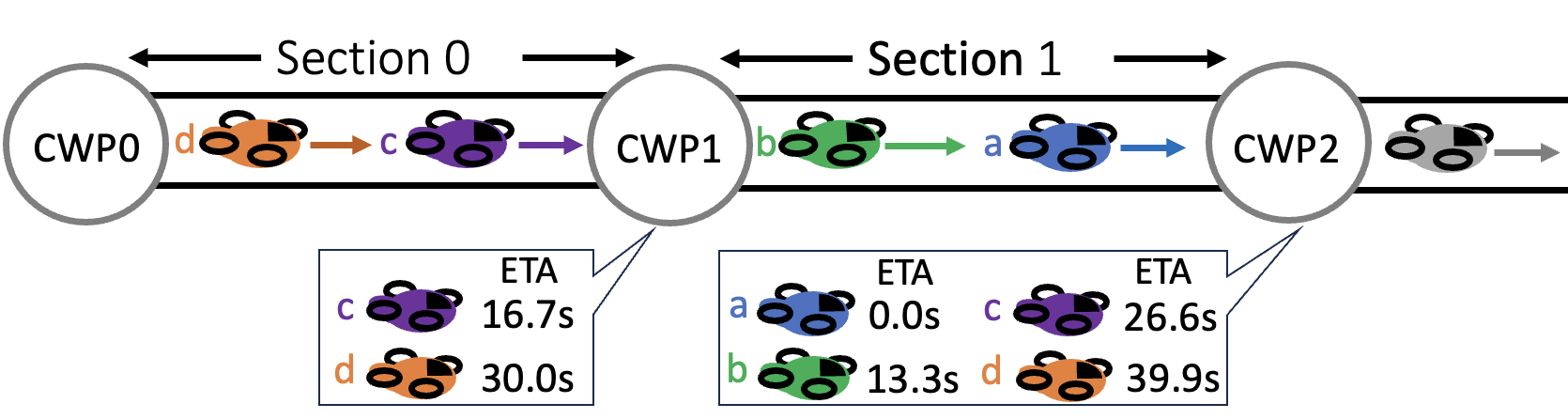}
\caption{A UAM corridor network with ETAs at CWPs.}
\label{fig:UAM_ETA}
\end{figure}

\subsection{Problem Formulation}

According to Near Mid-Air Collision (NMAC) avoidance rules, a minimum separation
distance is required to ensure safety between aerial vehicles, typically set to a
fixed length~\citep{johnson2017exploration, muna2021air}. 
We adopt this separation distance as the safety constraint for ETA 
scheduling at CWPs, requiring that the spacing between vehicles
never falls below $\mathit{safeD}$, and formulate the problem of finding a safe ETA gap $\Delta t_{{i, i+1}}$ as follows.
 
\begin{probl} \label{probl:etagap}
For two consecutive vehicles $i$ and $i{+}1$ entering the corridor network in order at CWP$_0$,
find 
\[
\Delta t_{{i, i+1}} \;:=\; t_{i+1,\text{enter}} - t_{i,\text{enter}} \;> 0
\]
\begin{align*} 
\begin{split}
&\text{such that }   x_i(t) 
      - x_{i+1}(t)  
\;\ge\; \mathit{safeD},\\[4pt]
&\quad\text{for all } t \in [
t_{i+1,\text{enter}},
t_{i,\text{enter}} + \tau_0 + \tau_1 + \cdots + \tau_{m-1}],
\end{split}
\end{align*}
where $t_{i,\text{enter}}$ and $t_{i+1,\text{enter}}$ are the entry
times at CWP$_0$ if vehicle $i$ and $i+1$, respectively, and $\mathit{safeD}$ is the NMAC minimum separation distance.
\end{probl}

Note that because the travel times $\tau_0, \tau_1, \dots, \tau_{m-1}$ are fixed, an ETA gap $\Delta t_{i,i+1}$ at CWP$_0$ also induces ETA gaps $\Delta t_{i,i+1} + \tau_0 + \tau_1 + \cdots + \tau_{j-1}$ at each CWP$_j$.



A trivial solution to Problem~\ref{probl:etagap} is to choose
$\Delta t_{i,i+1} > \tau_0 + \tau_1 + \cdots + \tau_{m-1}$,
which ensures that vehicle $i{+}1$ enters the corridor at CWP$_0$ after
vehicle $i$ has exited at CWP$_m$.  
For example, for the corridor in Fig.~\ref{fig:bounds}, this trivial solution requires $\Delta t_{i,i+1} > \tau_0 + \tau_1 + \tau_2 + \tau_3 = 42.1$\,s.
Although always safe, this choice is excessively conservative.  
We therefore aim for the smallest $\Delta t_{i,i+1}$ that still guarantees
$x_i(t)-x_{i+1}(t)\ge \mathit{safeD}$ for all admissible trajectories.

In the next section, we propose a computationally efficient method to compute
a sufficient ETA gap.  
Although this gap is not necessarily the globally minimal ETA gap, it is guaranteed to satisfy the safety condition in Problem~\ref{probl:etagap}
and is typically much smaller than the trivial bound above.







\section{Spatiotemporal Trajectory Bounds}
\label{sec:bounds}

\begin{figure}[t]
\centering
\includegraphics[width=1\columnwidth]{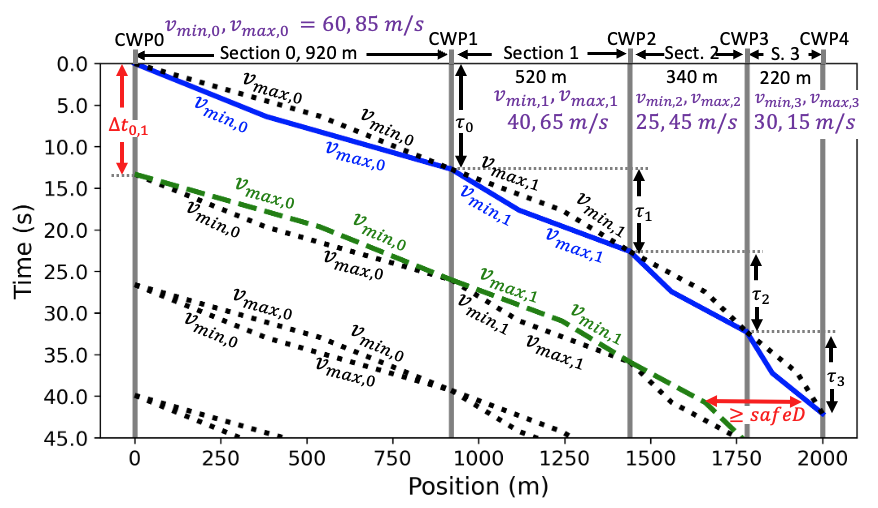}
\caption{Conservative spatiotemporal bounds of vehicles trajectories in  UAM corridors with $\Delta t_{{i, i+1}}=13.3$~s for $\mathit{safeD}=300$~m.
Length of corridors 0-3: 920, 520, 340, and 220~m; 
Speed limits: 
${v}_{\min,0}=60$, ${v}_{\max,0}=85$;
${v}_{\min,1}=40$, ${v}_{\max,1}=65$;
${v}_{\min,2}=25$, ${v}_{\max,2}=45$;
${v}_{\min,3}=15$, ${v}_{\max,3}=30$~m/s.
Travel times: 
$\tau_0=12.7$, $\tau_1=9.9$, $\tau_2=9.7$, $\tau_3=9.8$~s.
}
\label{fig:bounds}
\end{figure}


To compute a sufficient ETA gap for Problem~\ref{probl:etagap}, we first construct conservative
spatiotemporal bounds on vehicle trajectories, based on the corridor speed limits.
These bounds enable a computationally efficient method for deriving a guaranteed-safe ETA gap in the next section.

Given that all vehicles operate within the speed range $[\vMinj, \vMaxj]$ and spend $\tau_j$ time in each corridor section $j$, their admissible
spatiotemporal trajectories are bounded by the slopes as shown in
Fig.~\ref{fig:bounds}.  
These bounds arise from two extreme speed profiles: one that switches from
$\vMinj$ to $\vMaxj$ (blue solid line for vehicle~0 in
Fig.~\ref{fig:bounds}) and one that switches from $\vMaxj$ to $\vMinj$
(green dashed line for vehicle~1), assuming instantaneous speed changes.

Formally, let $l_j$ denote the length of corridor section
$j \in \{0, 1, \ldots, m-1\}$, and let CWP$_0$ be located at position~0.
We define the cumulative time and position offsets as
\begin{equation}
\label{eq:Lj Tij}
L_j := \sum_{\ell=0}^{j-1} l_{\ell} 
\quad\text{  and  }\quad
T_{i,j} :=  t_{i,\text{enter}} + \sum_{\ell=0}^{j-1} \tau_{\ell},  
\end{equation}
for $j \in \{1, 2, \ldots, m\}$,
with $L_0 = 0$ and $T_{i,0} = t_{i,\text{enter}}$.
Hence, $L_j$ is the position of CWP$_j$ and $T_{i,j}$ is the scheduled time at which vehicle $i$ arrives at CWP$_j$.


For each corridor section $j$, define $\widecheck{t}_j$ as the
unique nonnegative solution of
\begin{equation}\label{eq:tmintomax}
    v_{\min,j}\,\widecheck{t}_j
    + v_{\max,j}\bigl(\tau_j-\widecheck{t}_j\bigr)
    = l_j,
\end{equation}
representing the switching time of a trajectory that first flies at
$v_{\min,j}$ and then at $v_{\max,j}$. Note that such a solution exists if 
$v_{\min,j}\tau_j \le l_j \le v_{\max,j}\tau_j$, which is assumed throughout. 
For $t \in [T_{i,j},\,T_{i,j+1}]$, the trajectory $\widecheck{x}_i(t)$ is
\begin{align}
\begin{split}
&\widecheck{x}_i(t)  =
\begin{cases}
L_j + v_{\min,j}(t - T_{i,j}), \\
\qquad\qquad \text{if } T_{i,j} \le t \le T_{i,j} + \widecheck{t}_j, \\[6pt]
L_j + v_{\min,j}\widecheck{t}_j
  + v_{\max,j}(t - T_{i,j} - \widecheck{t}_j), \\
\qquad\qquad \text{if } T_{i,j} + \widecheck{t}_j \le t \le T_{i,j+1}.
\end{cases}
\end{split}
\label{eq:x_min_to_max}
\end{align}

Similarly, define $\widehat{t}_j$ as the unique nonnegative solution of
\begin{equation}\label{eq:tmaxtomin}
    v_{\max,j}\,\widehat{t}_j
    + v_{\min,j}\bigl(\tau_j-\widehat{t}_j\bigr)
    = l_j,
\end{equation}
corresponding to a trajectory that first flies at $v_{\max,j}$ and then 
$v_{\min,j}$.  
The trajectory $\widehat{x}_i(t)$ is defined analogously by swapping
$v_{\min,j}$ and $v_{\max,j}$, and replacing $\widecheck{t}_j$ with
$\widehat{t}_j$ in~\eqref{eq:x_min_to_max}.




We now show that these two trajectories bound all admissible motions of a
vehicle that satisfy the speed limits~\eqref{eq vlim} and the prescribed
section travel times~$\tau_j$.

\begin{theorem}
\label{thm:bounds}
Consider any vehicle $i$ that satisfies the speed limits condition~\eqref{eq vlim}, and spends  $\tau_j$ time units in
each corridor section $j \in \{0,\dots,m-1\}$.
Then, for all $t \in [T_{i,0}, T_{i,m}]$,
the vehicle position $x_i(t)$ satisfies
\begin{equation*}
    \widecheck{x}_i(t)
\;\le\;
x_i(t)
\;\le\;
\widehat{x}_i(t),
\end{equation*}
i.e., every admissible trajectory lies between the two extreme spatiotemporal bounds.
\end{theorem}
\begin{pf} 
Fix a corridor section $j$ and consider any admissible trajectory $x_i(t)$
that satisfies the speed bounds~\eqref{eq vlim} and the fixed travel time
$\tau_j$.
Let $t^*$ be any time such that vehicle $i$ is flying in corridor section $j$, and define
$x^* := x_i(t^*)$.
Then, we have $t^* \in [T_{i,j},\,T_{i,j+1}]$ and $x^* \in [L_j,\,L_{j+1}]$.
To prove the theorem, it suffices to show that $x^*$ lies between the
lower and upper bounds
$\widecheck{x}_i(t^*)$ and $\widehat{x}_i(t^*)$.

Assume, for the sake of contradiction, that either
$x^* < \widecheck{x}_i(t^*)$ or
$\widehat{x}_i(t^*) < x^*$.
We treat only the former case, as the latter follows symmetrically.
Suppose $x^* < \widecheck{x}_i(t^*)$, and let $\widecheck{t}_j$ be the speed-switching time  in~\eqref{eq:tmintomax}.

If $t^* \le T_{i,j} + \widecheck{t}_j$, then from~\eqref{eq:x_min_to_max}, we have
$
\widecheck{x}_i(t^*)
= L_j + v_{\min,j}(t^* - T_{i,j}).
$
Since $x^* < \widecheck{x}_i(t^*)$, we obtain
\[
\frac{x^* - L_j}{t^* - T_{i,j}}
<
\frac{\widecheck{x}_i(t^*) - L_j}{t^* - T_{i,j}}
= v_{\min,j},
\]
implying an average speed strictly below $v_{\min,j}$ from $(L_j,T_{i,j})$ to $(x^*,t^*)$, contradicting~\eqref{eq vlim}.



Otherwise, $t^* > T_{i,j}+\widecheck{t}_j$. 
From the second case of~\eqref{eq:x_min_to_max},
\[
\widecheck{x}_i(t^*)
 = L_j + v_{\min,j}\widecheck{t}_j
   + v_{\max,j}(t^* - T_{i,j} - \widecheck{t}_j).
\]
Using 
$L_{j+1}
 = L_j + v_{\min,j}\widecheck{t}_j
     + v_{\max,j}(\tau_j - \widecheck{t}_j)$ and $T_{i, j+1} = T_{i,j} + \tau_j$,
this expression can be rewritten as
\[
\widecheck{x}_i(t^*)
 = L_{j+1} - v_{\max,j}(T_{i,j+1}-t^*).
\]
Since $x^* < \widecheck{x}_i(t^*)$, we obtain
\[
L_{j+1} - x^*
  >  L_{j+1} - \widecheck{x}_i(t^*) = v_{\max,j}(T_{i,j+1}-t^*),
\]
which requires an average speed strictly above $v_{\max,j}$ from
 $(x^*,t^*)$ to $(L_{j+1},\,T_{i,j+1})$, contradicting~\eqref{eq vlim}.



As neither case is possible, 
$\widecheck{x}_i(t^*) \le x^* \le \widehat{x}_i(t^*)$
for all $t^*$ in section $j$. Since $T_j$ and $L_j$ accumulate across sections, the bound holds over the
entire corridor network.
\hfill\hfill\qed\end{pf}
\begin{rmk}
The bounds in Theorem~\ref{thm:bounds} apply to admissible trajectories
satisfying the section speed limits and prescribed section travel times.
Additional jerk, actuator, or transition-dynamics limits preserve the safety
argument if the realized trajectories remain within these bounds. If such
constraints or ETA-tracking errors change the realized CWP times, the ETA
schedule should be recomputed or enlarged by a timing buffer, e.g., a
$\pm\varepsilon$ ETA error at CWP$_0$ can be handled by adding
$2\varepsilon$ to the ETA gap.
\end{rmk}

\section{Computing a Sufficient ETA Gap}
\label{sec:etagap} 

This section derives a computationally efficient sufficient ETA gap.
By Theorem~\ref{thm:bounds},
$\widecheck{x}_i(t)$ lower-bounds the leader trajectory and
$\widehat{x}_{i+1}(t)$ upper-bounds the follower trajectory.
Therefore, $\widecheck{x}_i(t)-\widehat{x}_{i+1}(t)$ lower-bounds their
separation, yielding the following sufficient condition.

\begin{lem}
\label{lem:etagap_bound_sufficient}
Consider two consecutive vehicles $i$ and $i+1$ satisfying the speed-limit
condition~\eqref{eq vlim} and spending $\tau_j$ time units in each corridor
section $j\in\{0,\dots,m-1\}$. Let $\Delta t_{i,i+1}=T_{i+1,0}-T_{i,0}$ denote their ETA gap at CWP$_0$.
If
\begin{equation}
\label{eq:bound_sep_condition}
\widecheck{x}_i(t)-\widehat{x}_{i+1}(t)\ge \mathit{safeD},
\quad \forall t\in [T_{i+1,0},T_{i,m}],
\end{equation}
then $\Delta t_{i,i+1}$ is a feasible solution of Problem~\ref{probl:etagap}.
\end{lem}
 
\begin{pf}
By Theorem~\ref{thm:bounds}, for all $t\in [T_{i+1,0},T_{i,m}]$,
$x_i(t)\ge \widecheck{x}_i(t)$ and
$x_{i+1}(t)\le \widehat{x}_{i+1}(t)$. Hence,
\[
x_i(t)-x_{i+1}(t)
\ge
\widecheck{x}_i(t)-\widehat{x}_{i+1}(t)
\ge \mathit{safeD}.
\]
Therefore, the condition in Problem~\ref{probl:etagap} holds.
\hfill\hfill\qed
\end{pf}
 
Next, we compute the smallest ETA gap satisfying \eqref{eq:bound_sep_condition}. 
To do so, we extract the \emph{critical points} of the trajectory bounds 
$\widecheck{x}_i(t)$ and $\widehat{x}_{i+1}(t)$, which correspond to the
times at which the slope changes.
Let $\widecheck{C}_i$ and $\widehat{C}_{i+1}$ denote these sets of 
critical points, defined as follows:
\begin{align*}
\begin{alignedat}{2}
\widecheck{C}_i
&\;:=\;&\;
&\bigcup_{j=0}^{m-1} \Bigl\{
   (L_j,\, T_{i,j}),\,
   (L_j + v_{\min,j}\,\widecheck{t}_j,\, T_{i,j} + \widecheck{t}_j), \\
&&&
   (L_{j+1},\, T_{i,j+1})
\Bigr\},
\\[8pt]
\widehat{C}_{i+1}
&\;:=\;&\;
&\bigcup_{j=0}^{m-1} \Bigl\{
   (L_j,\, T_{i+1,j}),\,
   (L_j + v_{\max,j}\,\widehat{t}_j,\\
&&&   T_{i+1,j} + \widehat{t}_j), \, 
   (L_{j+1},\, T_{i+1,j+1})
\Bigr\},
\end{alignedat}
\end{align*}
where  $\widecheck{t}_j$ and $\widehat{t}_j$ are the unique
nonnegative solutions of \eqref{eq:tmintomax} and \eqref{eq:tmaxtomin},
respectively

For a given pair of vehicles $i$ and $i{+}1$,  
we collect all time instants at which either trajectory bound changes slope.
Define the set of relevant time instants as
\begin{align}
\label{eq:critical_time_set}
\widetilde{T_i}
:= \{\, t : \exists x, (x,t)\in
\widecheck{C}_i \cup \widehat{C}_{i+1} \,\}
\cap
[T_{i+1,0},\, T_{i,m}].
\end{align}

The following theorem states that \eqref{eq:bound_sep_condition} can be ensured by checking the  separation only at the time instants in $\widetilde{T_i}$.

\begin{theorem}
\label{thm:critical points}
Given a pair of vehicles $i$ and $i{+}1$ and a safety distance $\mathit{safeD} \ge 0$.
\[
\begin{aligned}
&\text{If }
\begin{aligned}[t]
\widecheck{x}_i(t)
- \widehat{x}_{i+1}(t)
&\;\ge\; \mathit{safeD} \text{ for all } t \in \widetilde{T_i},
\end{aligned}
\\ 
&\text{Then }
\begin{aligned}[t]
\widecheck{x}_i(t)
- \widehat{x}_{i+1}(t)
&\;\ge\; \mathit{safeD}  
 \text{ for all } t \in [T_{i+1,0},\,T_{i,m}].
\end{aligned}
\end{aligned}
\]
\end{theorem}
 
\begin{pf}
Let 
$d(t) := \widecheck{x}_i(t) - \widehat{x}_{i+1}(t)$.
From~\eqref{eq:x_min_to_max}, both bounds
$\widecheck{x}_i(t)$ and $\widehat{x}_{i+1}(t)$ are continuous
piecewise-affine functions whose slope changes 
only at the
critical times in $\widetilde{T_i}$.  
Hence, $d(t)$ is also continuous and affine on each interval between
consecutive elements of $\widetilde{T_i}$.
Order the elements of $\widetilde{T_i}$ as
$T_{i+1,0} = t_0 < t_1 < \dots < t_K = T_{i,m}$.
On each interval $[t_k, t_{k+1}]$, the function $d(t)$ is affine
and thus attains its minimum at one of the endpoints.
By the assumption of the theorem, $d(t_k) \ge \mathit{safeD}$ for all $k$, 
which implies 
$d(t) \ge \mathit{safeD}$ for all $t \in [t_k, t_{k+1}]$.
Taking the union over all intervals yields
$d(t) \ge \mathit{safeD}$ for all 
$t \in [T_{i+1,0},\,T_{i,m}]$,
which proves the claim.\hfill\hfill\qed 
\end{pf}

By Lemma~\ref{lem:etagap_bound_sufficient} and Theorem~\ref{thm:critical points}, it suffices to enforce the separation constraint at the critical times in $\widetilde{T_i}$.
This yields the following finite-dimensional optimization problem.

\begin{probl}[Reduced ETA-Gap Optimization]
\label{probl:reduced_mingap}
\begin{align*}
\begin{split}
\underset{\Delta t_{i,i+1}\,\in\,[0,\;\sum_{j=0}^{m-1}\tau_j]}{\min}\quad
& \Delta t_{i,i+1} \\[3pt]
\text{subject to}\quad
& \widecheck{x}_i(t)-\widehat{x}_{i+1}(t)
    \;\ge\; \mathit{safeD},
    \quad \forall\, t \in \widetilde{T_i}, \\[4pt]
& T_{i+1,0}-T_{i,0} \;=\; \Delta t_{i,i+1}.
\end{split}
\end{align*}
\end{probl}

\begin{rmk}
Since all vehicles are assumed to share the same travel time $\tau_j$ in each
corridor section $j$, the resulting minimum ETA gap
$\Delta t_{i,i+1}$ is the same for all consecutive pairs.
We nevertheless formulate Problem~\ref{probl:reduced_mingap} for a generic pair
$(i,i{+}1)$ to highlight that the method naturally extends to heterogeneous
section travel times across vehicles.
\end{rmk} 

Combining Theorem~\ref{thm:critical points} with
Lemma~\ref{lem:etagap_bound_sufficient} yields the following conclusion.

\begin{coro}
\label{cor:combine}
Consider two consecutive vehicles $i$ and $i+1$ satisfying the speed-limit
condition~\eqref{eq vlim} and spending $\tau_j$ time units in each corridor
section $j\in\{0,\dots,m-1\}$. Let $\Delta t_{i,i+1}=T_{i+1,0}-T_{i,0}$ denote their ETA gap at CWP$_0$.
If $\Delta t_{i,i+1}$ is a solution of Problem~\ref{probl:reduced_mingap}, 
then $\Delta t_{i,i+1}$ is a feasible solution of Problem~\ref{probl:etagap}.
In particular, it guarantees
\[
x_i(t) - x_{i+1}(t) \;\ge\; \mathit{safeD},
\quad \forall t \in [T_{i+1,0},\,T_{i,m}].
\]
\end{coro}

Thus, using this ETA gap ensures that all vehicles remain at least
$\mathit{safeD}$ apart for all admissible trajectories.

\begin{rmk}
Corridor lengths and speed limits strongly affect the minimum ETA gap
$\Delta t_{{i, i+1}}$.
For instance, setting all corridor section lengths in Fig.~\ref{fig:bounds} to
500 m increases the gap to 16.5 s, compared with 13.3 s in the current setting.
This observation suggests that the ETA-gap computation can also serve as a
corridor-design evaluation tool. We leave systematic sensitivity analysis over
corridor lengths and speed-limit profiles for future work.
\end{rmk}

\noindent\textbf{Computation procedure.}
We end this section by summarizing the procedure for computing a sufficient ETA gap.
For a given safety distance $\mathit{safeD}$, the steps for each vehicle pair $(i,i+1)$ are:
\begin{enumerate}
    \item Compute the scheduled CWP times $T_{i,j}$ from the section travel
    times $\tau_j$ for all corridor sections $j$.
    \item Construct the spatiotemporal trajectory bounds
    $\widecheck{x}_i(t)$ and $\widehat{x}_{i+1}(t)$ using
    \eqref{eq:x_min_to_max} and its max-to-min counterpart.
        \item Form the critical-time set $\widetilde{T_i}$ from the critical points of
    the two trajectory bounds as in \eqref{eq:critical_time_set}.
    \item Solve Problem~\ref{probl:reduced_mingap}, i.e., find the smallest
    $\Delta t_{i,i+1}$ such that
    $\widecheck{x}_i(t)-\widehat{x}_{i+1}(t)\ge \mathit{safeD}$ for all
    $t\in\widetilde{T_i}$.
\end{enumerate}
By Corollary~\ref{cor:combine}, the resulting ETA gap  $\Delta t_{i,i+1}$ guarantees the required
separation for all admissible trajectories.

\section{Experimental Result}
\label{sec:simulation}

We validate the theoretical results through numerical simulation, comparing
operations with and without the proposed ETA gaps.  
We choose a decreasing-speed corridor as a challenging test case because downstream speed reduction can cause traffic compression and increase collision risk, e.g., when vehicles approach a vertiport. 
All simulations are implemented in Python~3.12 with a discrete-time step of
$\delta_t = 0.1$\,s. 
We use the corridor illustrated in Fig.~\ref{fig:bounds}, with $\mathit{safeD}$ values ranging from 100\,m to 1200\,m to assess performance under varying spacing requirements.
 
Vehicle motion follows the discrete-time version of \eqref{eq:kinematic}:
\[
x_i(k{+}1) = x_i(k) + \delta_t\, v_i(k),\quad
v_i(k{+}1) = v_i(k) + \delta_t\, a_i(k).
\]

For the first vehicle $i=0$, while it is in section~$j$, the acceleration is
chosen to reduce the difference between its current speed and the average speed
required to traverse the section:
$a_0(k)=l_j/\tau_j-v_0(k)$.
For each following vehicle $i>0$, inter-vehicle interactions are modeled using
the Helly vehicle-following model
\citep{helly1959simulation, ambrosio2018parameter, li2024helly}:  
\begin{equation*} 
    a_i(k)
  = \lambda_x \!\left( \Delta x - D(v_i(k)) \right) + \lambda_v\, \Delta v,
\end{equation*} 
where $\Delta x = x_{i-1}(k)-x_i(k)$,
$\Delta v= v_{i-1}(k)-v_i(k)$, and
\[
D(v_i(k)) = D_{\min} + T_{\mathrm{des}} v_i(k)
\]
is the desired spacing, with minimum distance $D_{\min}$ and dynamic spacing
$T_{\mathrm{des}}v_i(k)$.
All accelerations are adjusted as needed to keep speeds within
$[v_{\min,j},v_{\max,j}]$, and clipped to $[-3,2]$\,m/s$^2$, which are used as admissible acceleration bounds in the simulation.
We consider parameter values $\lambda_x = 0.7$, $\lambda_v = 0.5$,
$T_{\mathrm{des}} = 1.2\,\mathrm{s}$, and vary
$D_{\min} = \mathit{safeD} \in \{100,200,\ldots,1200\}$\,m.
The values of $\lambda_x$, $\lambda_v$, and $T_{\mathrm{des}}$ follow~\citep{li2024helly}.

We consider continuous vehicle entry over a 100\,s interval.
All vehicles enter the network at CWP$_0$ with an initial speed of
$82.5$\,m/s.  
Two operation modes are evaluated:

\emph{1) Without ETA gaps (no-ETA):}
Vehicles select their accelerations as described above.
The first vehicle enters at time zero, and vehicles $i>0$ enter the corridor
network at CWP$_0$ as soon as the resulting accelerations are positive.

\emph{2) With ETA gaps (ETA):}
Vehicles follow the same acceleration policy as in the no-ETA case, with
additional adjustments when needed to match the section travel times 
$\tau_0,\tau_1,\ldots,\tau_{m-1}$.
The first vehicle enters at time zero, and each following vehicle enters at
CWP$_0$ according to the minimum ETA-gap computed by solving
Problem~\ref{probl:reduced_mingap}.

\begin{figure}[t] 
\centering
\includegraphics[width=0.9\columnwidth]{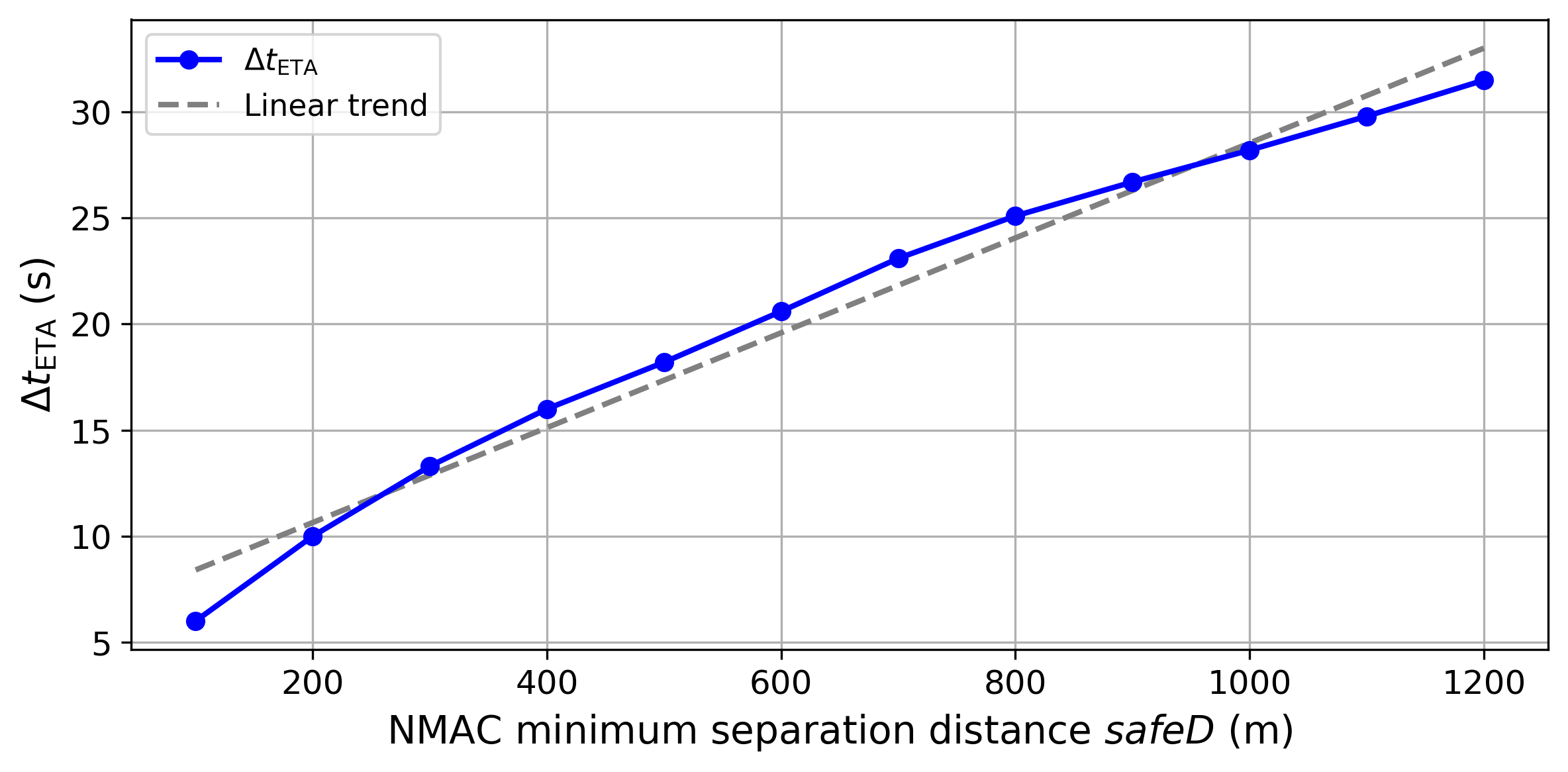}
\caption{Relationship between $\mathit{safeD}$ and   $t_{{i, i+1}}$. 
}
\label{fig:ex3}
\end{figure}
\begin{figure}[t] 
\centering
\includegraphics[width=0.9\columnwidth]{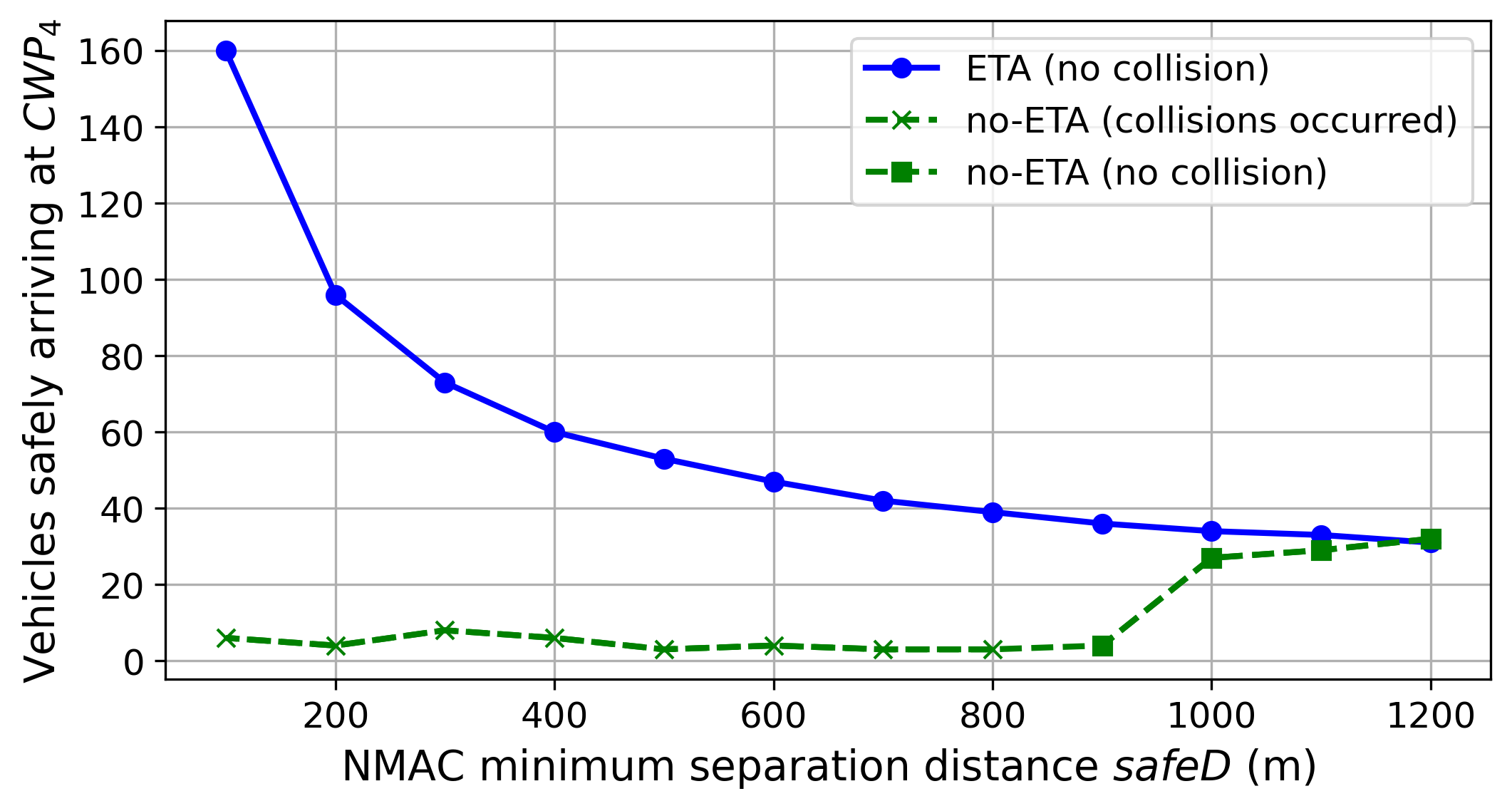}
\caption{Number of vehicles that safely arrived at $\text{CWP}_4$. 
}
\label{fig:ex2}
\end{figure}
\begin{figure}[t] 
\centering
\includegraphics[width=0.9\columnwidth]{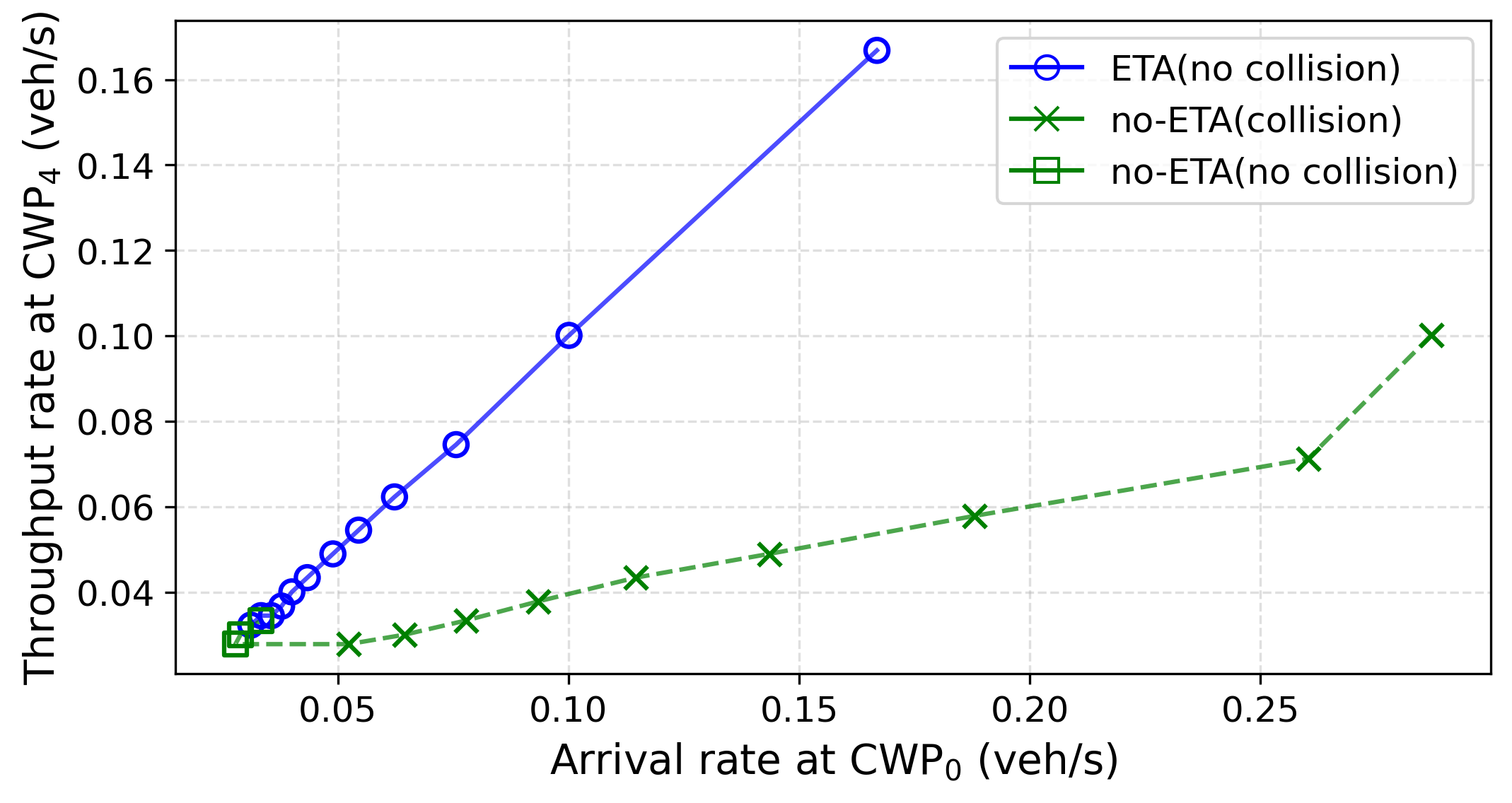}
\caption{Vehicle arrival rate at CWP$_0$ vesus the throughput rate at CWP$_4$.
}
\label{fig:ex4}
\end{figure}
\begin{figure}[t] 
\centering
\includegraphics[width=0.9\columnwidth]{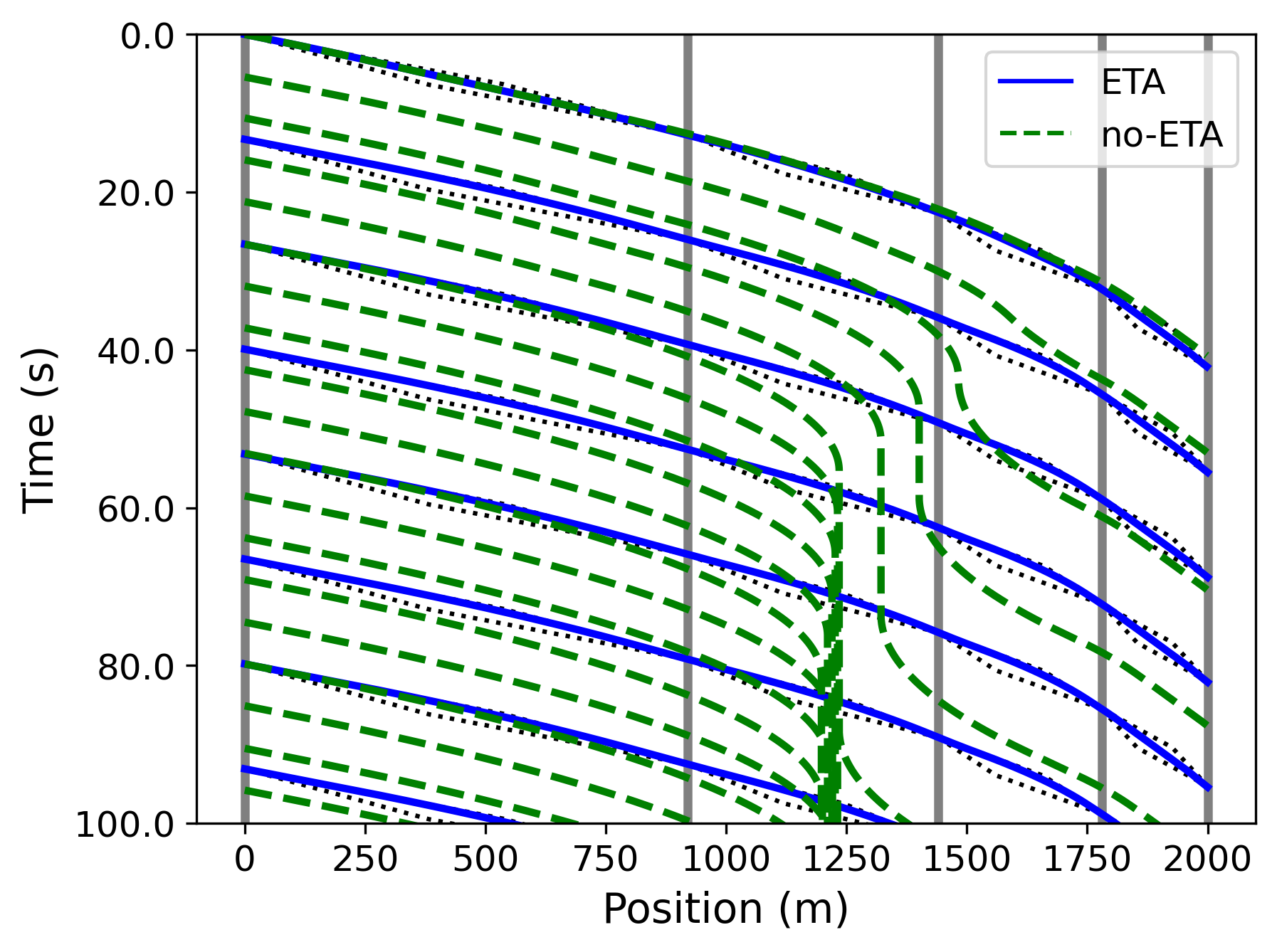}
\caption{Spatiotemporal trajectories of vehicles in UAM corridors for $\mathit{safeD}=300$ m, under the two operation modes: with and without ETA gaps. 
}
\label{fig:ex}
\end{figure}



\subsection*{Simulation Results and Analysis.}

We solve the optimization Problem~\ref{probl:reduced_mingap} for various values of
$\mathit{safeD}$.  
Across all tested values, the minimum ETA gap $\Delta t_{i,i+1}$ is computed in under $0.003$\,s on a MacBook Pro (M4 Max, 64\,GB) using the SciPy function \texttt{scipy.optimize.minimize\_scalar}.  
The resulting ETA gaps for $\mathit{safeD} = 100, 200, \ldots, 1200$\,m are
6.0, 10.0, 13.3, 16.0, 18.2, 20.6, 23.1, 25.1, 26.7, 28.2, 29.8, and
31.5\,s, respectively.  
As shown in Fig.~\ref{fig:ex3}, the required ETA gap increases monotonically
with $\mathit{safeD}$,  with growth gradually diminishing for larger values.

Figure~\ref{fig:ex2} shows the relation between $\mathit{safeD}$ and the number
of vehicles that successfully exit the corridor network at CWP$_4$ without collision.
The ETA-gap policy prevents all collisions for the entire range of
$\mathit{safeD}$ considered, even for $\mathit{safeD}=100$\,m.
In contrast, the no-ETA mode experiences collisions for all cases with
$\mathit{safeD} \le 800$\,m.  
For the ETA-gap mode, the number of vehicles that complete the full corridor
network within the simulation horizon (100\,s) decreases as $\mathit{safeD}$
increases, which is expected because larger separation distances require larger
ETA gaps at entry.
Nevertheless, for all $\mathit{safeD} \le 1100$\,m, the ETA-gap mode allows
more vehicles to exit safely than the no-ETA mode.  
Figure~\ref{fig:ex4} shows the arrival rate at CWP$_0$ versus the throughput rate at CWP$_4$. 
The ETA-gap mode achieves a nearly linear throughput increase, while the no-ETA mode saturates early and remains much lower. 
Thus, ETA-based coordination improves both safety and flow.
Moreover, across all simulations, every vehicle operating with ETA-gap remains at least
$\mathit{safeD}$ apart at all times.

Figure~\ref{fig:ex} shows example spatiotemporal trajectories with and without ETA gaps for
$\mathit{safeD}=300$\,m. 
Vehicles operating under the ETA-gap policy exhibit smooth, uniform flow: their
trajectories follow the bounds 
$\widecheck{x}_i(t)$ and $\widehat{x}_i(t)$.  
In contrast, vehicles in the no-ETA mode experience reactive braking as they
approach sections with reduced speed limits, resulting in 
congestion
and, eventually, collisions.



These results show that the ETA-gap mechanism preserves safe inter-vehicle separation while
substantially improving traffic flow compared to the unscheduled (no-ETA) mode.

\section{Conclusion}
\label{sec:conclusion}

We presented a computationally lightweight framework for safety-assured arrival scheduling in sequential UAM corridors.
By coordinating ETAs at CWPs, we derived a sufficient ETA-gap rule that guarantees inter-vehicle separation across corridor sections with heterogeneous speed limits using conservative spatiotemporal bounds.
Numerical simulations confirmed that the ETA-gap policy maintained the required spacing in all scenarios, whereas unscheduled operations experienced congestion and collisions.
These findings demonstrate that ETA coordination can improve safety and flow stability in UAM corridors.

These results provide a basis for further developing ETA-based coordination for UAM corridor operations.
Future work will consider additional separation criteria, corridor-design strategies, alternative admission policies, ETA-tracking uncertainty, and merging traffic.
Other interesting future directions include higher-fidelity vehicle dynamics, mixed speed-profile corridors, and sensitivity studies of vehicle-following parameters and corridor configurations.

\section*{Acknowledgment}
During the preparation of this work, the authors used ChatGPT-5.4 to support grammar, spelling checks, and wording refinement. 
After using this tool/service, the authors reviewed and edited the content as needed and takes full responsibility for the content of the publication. 

\bibliography{ifacconf}             
                                                   







\end{document}